\documentclass[runningheads]{llncs}
\usepackage{amsmath,amssymb,amsfonts}
\usepackage{graphicx}
\usepackage{mathtools}
\usepackage{float}
\newcommand{\qedwhite}{\hfill \ensuremath{\Box}}
\newtheorem{assumption}{Assumption}
\newcommand\cq{\mathcal{Q}}
\newcommand\cb{\mathcal{B}}
\newcommand\ck{\mathcal{K}}
\DeclareMathOperator{\lcm}{lcm}
\begin{document}
\title{Deception in Oligopoly Games via Adaptive Nash Seeking Systems\thanks{Supported by NSF}}
\author{Michael Tang\inst{1} \and
Miroslav Krstic\inst{2} \and
Jorge Poveda\inst{1}}
\authorrunning{M. Tang et al.}
\institute{
Department of Electrical and Computer Engineering, University of California-San Diego, La Jolla, California, USA\\
\email{\{myt001, jipoveda\}@ucsd.edu}\and Department of Mechanical and Aerospace Engineering, University of California-San Diego, La Jolla, California, USA\\
\email{mkrstic@ucsd.edu}}
\maketitle              % typeset the header of the contribution
\begin{abstract}
In the theory of multi-agent systems, deception refers to the strategic manipulation of information to influence the behavior of other agents, ultimately altering the long-term dynamics of the entire system. Recently, this concept has been examined in the context of model-free Nash equilibrium seeking (NES) algorithms for noncooperative games \cite{tang2024deception}. Specifically, it was demonstrated that players can exploit knowledge of other players' exploration signals to drive the system toward a ``deceptive" Nash equilibrium, while maintaining the stability of the closed-loop system. To extend this insight beyond the duopoly case, in this paper we conduct a comprehensive study of deception mechanisms in  \emph{$N$-player oligopoly markets}. By leveraging the structure of these games and employing stability techniques for nonlinear dynamical systems, we provide game-theoretic insights into deception and derive specialized results, including stability conditions. These results allow players to systematically adjust their NES dynamics by tuning gains and signal amplitudes, all while ensuring closed-loop stability. Additionally, we introduce novel sufficient conditions to demonstrate that the (practically) stable equilibrium point of the deceptive dynamics corresponds to a true Nash equilibrium of a different game, which we term the ``deceptive game." Our results show that, under the proposed adaptive dynamics with deception, a victim firm may develop a distorted perception of its competitors' product appeal, which could lead to setting suboptimal prices.

\keywords{Nash equilibrium seeking  \and Deception \and Oligopoly.}
\end{abstract}
\section{Introduction}
\subsection{Motivation}
The study of multi-agent systems (MAS) is gaining increasing significance, particularly in engineering fields such as smart grids, robotics, and machine learning. Game theory offers a robust framework for analyzing the strategic interactions between rational decision-makers in such systems. However, traditional game-theoretic concepts, such as Nash equilibrium \cite{540b73bd-a3f1-333e-a206-c24d0fbbb8bc}, may need to be revisited in scenarios where agents have access to \emph{privileged information}. In the context of learning in games \cite{young2004strategic}, such asymmetry of information can be exploited to manipulate the system's long-term behavior, leading to outcomes that benefit some agents while disadvantaging others. For instance, in the context of Nash equilibrium seeking in non-cooperative games \cite{bacsar1998dynamic,6060862}, privileged agents can exploit knowledge of other agents' exploration policies to interfere with their learning processes, without altering their own learning capabilities. When this interference maintains the overall system’s stability, it can cause the naive agents to converge to incorrect steady-state models or beliefs . This phenomenon, known as \emph{deception in games}, has received significant attention during the last years due to its potential implications in the context of cyber-security and resilient decision making in socio-technical systems. Algorithmic deception has been studied across various domains, including robotics and aerospace control \cite{dec_thesis,ho2022game,dragan2015deceptive}.  Similarly, studies such as \cite{wagner2011acting} propose algorithms for robots to decide when to deceive, as illustrated in hide-and-seek experiments. While deception can aid robots in achieving particular goals in non-competitive environments, it has also been explored in competitive scenarios, such as signaling games  \cite{kouzehgar2019fuzzy}. These concepts are further explored in works investigating deception in multi-agent systems, offering strategies to counter deceptive signals and attacks \cite{5339236}. Deception has also been studied in the context of biological systems \cite{smith1987deception} and societal systems \cite{mitchell1986deception}.

In this paper, we focus on studying deception in Nash equilibrium-seeking (NES) problems within non-cooperative games. In these scenarios, a finite number of agents, or players, seek to maximize their individual profits, which depend not only on their own actions but also on the actions of others. Given the challenge associated with computation of Nash equilibria \cite{540b73bd-a3f1-333e-a206-c24d0fbbb8bc} in noncooperative games, the study of NES algorithms has become a very hot research topic. Various algorithms have been designed, including distributed \cite{ye2017distributed,YI2019111,ye2023distributed}, semi-decentralized \cite{belgioioso2017semi,zou2022semidecentralized} and hybrid \cite{9968117,wang2020distributed} algorithms. However, in some cases, the agents do not have precise knowledge of their cost functions, and hence they must rely on adaptive seeking dynamics that incorporate exploration and exploitation strategies. To address this, extremum-seeking based NES dynamics were introduced in \cite{6060862} and have been extended to stochastic settings \cite{doi:10.1137/100811738}, systems with delays \cite{oliveira2021nash}, nonsmooth algorithms \cite{9760031}, etc. In this setting, agents with privileged knowledge of the exploration policy used by others can manipulate their own exploration policy to induce false beliefs in the other agents, making them take actions that are detrimental for them. Such type of deception was recently introduced and studied in \cite{tang2024deception} for a general class of games, establishing conditions that preserve stability in the overall system. In particular, it was shown that when players in a non-cooperative game implement the model-free Nash equilibrium seeking (NES) scheme from \cite{6060862}, a player who gains insight into another player's exploration policy can manipulate the Nash equilibrium to their advantage. The proposed deception mechanism involves an additive and dynamic modification to the deceiver’s action update, incorporating the victim’s exploration frequency. This modification is adjusted using first or second-order dynamics. In the context of duopoly games, deceptive players have been shown to manipulate the victim's perception of their sales function. Using an averaging and singular perturbation approach, it was further shown that the proposed deception mechanisms preserve stability in non-cooperative games with general nonlinear payoffs, including strongly monotone games. While the results in \cite{tang2024deception} were the first to establish dynamic deception through model-free NES dynamics, their general applicability is constrained by certain assumptions. For example, the ``stability-preserving set" described in \cite{tang2024deception}, a key element in the singular perturbation analysis, is only guaranteed to cover a small neighborhood around a specific point. Additionally, although the work in \cite{tang2024deception} relaxes the diagonal dominance condition from \cite{6060862}, it does so at the cost of requiring all players to use identical gains and amplitudes in their NES dynamics.

% \subsection{Related Works}
% Given the challenge associated with computation of Nash equilibria \cite{540b73bd-a3f1-333e-a206-c24d0fbbb8bc} in noncooperative games, the study of NES algorithms has become a very hot research topic. Various algorithms have been designed, including distributed \cite{ye2017distributed,YI2019111,ye2023distributed}, semi-decentralized \cite{belgioioso2017semi,zou2022semidecentralized} and hybrid \cite{9968117,wang2020distributed} algorithms. The extremum-seeking based NES dynamics introduced in \cite{6060862} have been extended to stochastic settings \cite{doi:10.1137/100811738}, systems with delays \cite{oliveira2021nash}, nonsmooth algorithms \cite{9760031}, etc. Some fundamental assumptions in NES algorithms are that the players share truthful information with each other and the distribution of information is symmetric.
\subsection{Our Contributions}
In this paper, we introduce several new results that relax some of the previous assumptions considered in the literature \cite{tang2024deception} , and, additionally, we introduce new results and computations in the context of general $N$-player  $N$-player oligopoly markets, such as those studied in \cite{6060862}. Furthermore, we exploit the structure of the oligopoly market to derive sharper stability results and characterizations that quantify the influence of deception in this context. We also provide a broader estimate for the stability-preserving set in the nominal average dynamics and offer results that determine when the new equilibrium point under deception is truly a Nash equilibrium, rather than just an equilibrium, for the deceptive game. This framework further allows us to consider a more general class of NES dynamics, where players are permitted to use different gains and amplitudes in their update laws. The effectiveness of these results are also demonstrated through numerical simulations.
\section{System Model and Oligopoly Formulation}
In this section, we describe the types of games considered in this paper, as well as the model-free deception dynamics.
\subsection{NES for the oligopoly}
Consider an $N$-player noncooperative game, where player $i$ implements action $x_i\in\mathbb{R}$ and aims to unilaterally minimize their cost function $J_i:\mathbb{R}^N\to\mathbb{R}$. We use $[N]:=\{1,2,...,N\}$ to denote the set of players, and we let $x=[x_1,...,x_N]^\top$ denote the vector of players' actions. Similarly we use $x_{-i}\in\mathbb{R}^{N-1}$ to denote the vector of all players' actions except for the action of player $i$. Given real-valued cost functions $J_i(x_i,x_{-i}):\mathbb{R}^N\to\mathbb{R}$, for all $i$, a policy $x^*\in\mathbb{R}^N$ is called a \emph{Nash equilibrium} if it satisfies
\begin{equation}\label{eqNE}
x_i^*=\text{arg}\min_{x_i}J_i(x_i,x_{-i}^*),~~\forall~i\in [N].
\end{equation}
We define the \emph{pseudogradient} of the game to be $\mathcal{G}(x):=[\nabla_1 J_1(x),...,\nabla_N J_N(x)]^\top$ where $\nabla_i J_i(x)$ is the partial derivative of $J_i(x)$ with respect to $x_i$. Given $v\in\mathbb{R}^N$ we use $\text{diag}(v)\in\mathbb{R}^{N\times N}$ to denote the diagonal matrix with $i$-th diagonal element given by $v_i$. Since we focus on oligopoly games where each player \( i \) controls the price \( x_i \) of their own product, the cost functions of interest take the form \cite{6060862}:
\begin{equation}
    J_i(x(t))=-s_i(x(t))(x_i(t)-m_i)\label{jcost}
\end{equation}
where $m_i$ is the marginal cost of the product generated  by the $i^{th}$ player, and $s_i$ is their sales function, which is given by
\begin{equation}
    s_i(x(t))=\frac{R_{\|}}{R_i}\left(S_d-\frac{x_i(t)}{\overline{R}_i}+\sum_{j\neq i}^{N}\frac{x_j(t)}{R_j}\right)\label{salesi}.
\end{equation}
Here, $S_d$ is the total consumer demand and $R_i$ represents the resistance" that consumers have towards buying the product offered by firm $i$. In other words, the desirability of product $i$ is inversely proportional to $R_i$. The quantities $R_{\|}$ and $\overline{R}_i$ are given by
\begin{align}
        \frac{1}{R_{\|}}=\sum_{k=1}^N\frac{1}{R_k},~~~\text{and}~~\quad
        \frac{1}{\overline{R}_i}=\sum_{k\neq i}^N\frac{1}{R_k}.
\end{align} 
As in \cite{6060862}, the sales function is motivated by an electrical circuit analogy, illustrated in Figure \ref{circ_fig}.
\begin{figure}[H]
\centering
\includegraphics[width=0.55\textwidth]{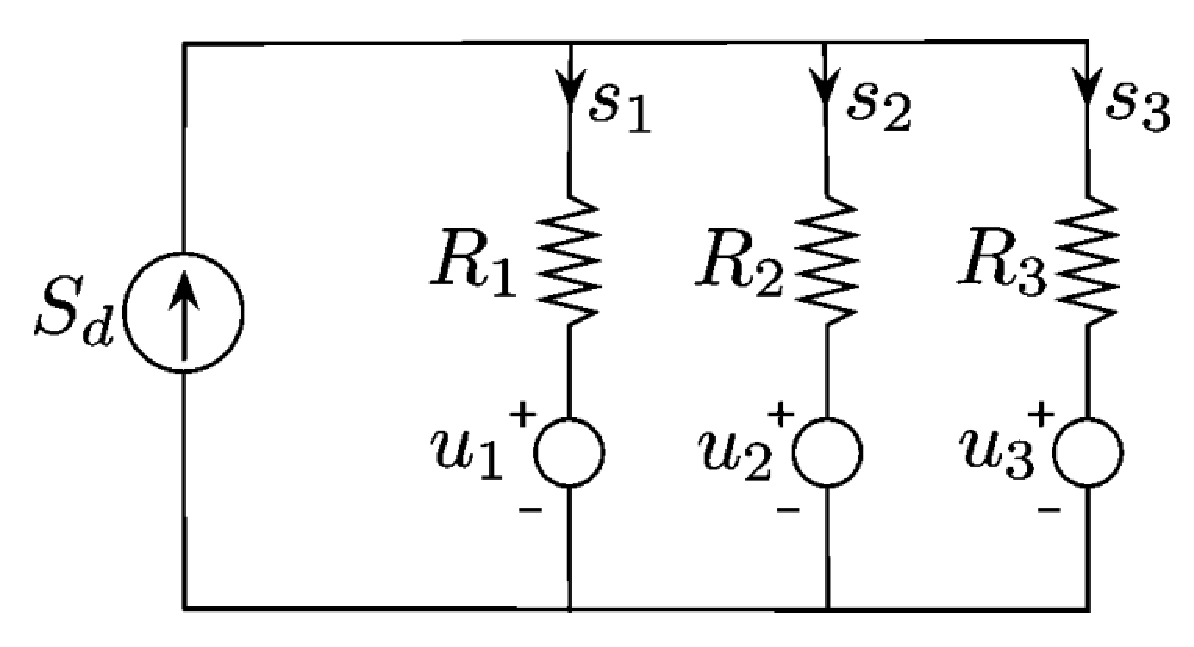}
\caption{Market-circuit analogy borrowed from \cite{6060862}, which models the sales $s_1, s_2, s_3$ of a three-player oligopoly as currents in a 3-resistor parallel circuit.} \label{circ_fig}
\end{figure}
In Figure \ref{circ_fig}, the total demand $S_d$ can be thought of as a current generator, the prices $x_i$ are voltage sources, and $R_i$ is the resistance towards product $i$. It is also important to note that $s_i(x)(x_i-m_i)$ represents the profit of firm $i$, but we define $J_i$ as \eqref{jcost} since, without loss of generality, we see $J_i$ as a \emph{cost} to be \emph{minimized} by the $i^{th}$ player. 

To converge to a neighborhood of a ``standard'' Nash equilibrium by only using measurements of their cost $J_i$, players can implement the following model-free NES dynamics introduced in \cite{6060862}:
\begin{subequations}\label{extsc0}
 \begin{align}
     x_i(t)&=u_i(t)+a_i\sin(\omega_i t)\label{extscx}\\
     \dot{u}_i(t)&=-\frac{2k_i}{a_i} J_i(x(t))\sin(\omega_i t),\label{extscu}
 \end{align}   
\end{subequations}
where the gains $a_i, k_i>0$ are positive tunable parameters, and the frequencies $\omega_i\in\mathbb{R}_{>0}$ are selected to satisfy the following assumption:
\begin{assumption}\label{assumpw}
The frequencies satisfy $\omega_i\neq \omega_j$ for $i\neq j$, and $\omega_i=\omega\overline{\omega}_i$, where $\omega\in\mathbb{R}_{>0}$ and $\overline{\omega}_i\in\mathbb{Q}_{>0}, \ \forall i\in[N]$.
\end{assumption}
%
%We note that \eqref{extsc0} is a slight generalization of the NES dynamics considered in \cite{tang2024deception} since we allow the players to use distinct gains $k_i$ and amplitudes $a_i$. We allow this since the pseudogradient of the oligopoly satisfies the diagonal dominance condition needed in \cite{6060862}. 
It is easy to verify that the costs $J_i$ can be represented in the quadratic form
\begin{equation}
    J_i(x)=\frac12 x^\top Q_i x+b_i^\top x+\frac{R_{\|}S_d m_i}{R_i},
\end{equation}
where $Q_i\in\mathbb{R}^{N\times N}$ and $b_i\in\mathbb{R}^N$ have entries given by
\begin{subequations}
    \begin{align}
        [Q_i]_{jk}&=\begin{dcases}
            \frac{2R_{\|}}{R_i \overline{R}_i},\quad &\text{if } j=k=i\\
            -\frac{R_{\|}}{R_j R_k},\quad &\text{if } j\neq k \text{ and either } j=i \text{ or } k=i\\
            0,\quad &\text{else}
        \end{dcases}\label{qi}\\
        [b_i]_k&=\begin{dcases}
            -\frac{m_i R_{\|}}{R_i\overline{R}_i}-\frac{S_d R_{\|}}{R_i},\quad &\text{if } k=i\\
            \frac{m_i R_{\|}}{R_i R_k},\quad &\text{else}
        \end{dcases}
    \end{align}
\end{subequations}
Here, $[Q_i]_{jk}$ and $[b_i]_k$ denote the $(j,k)$ entry of $Q_i$ and $k$-th entry of $b_i$ respectively. It follows immediately that the pseudogradient of the oligopoly can be represented in the form $\mathcal{G}(x)=\mathcal{Q}x+\mathcal{B}$, where $\cq\in\mathcal{R}^{N\times N}$ and $\mathcal{B}\in\mathbb{R}^N$ satisfy $[\cq]_{i:}=[Q_i]_{i:}$ and $[\cb]_i=[b_i]_i$ for all $i$. Here, $[E]_{i:}$ denotes row $i$ of matrix $E$.
\subsection{Deception in the Oligopoly}
To incorporate deception into the NES dynamics \eqref{extsc0}, in this paper we consider the following modified model-free NES algorithm:
\begin{definition}
Player $i\in[N]$ is said to be deceptive towards a set of players $\mathcal{D}_i\subset [N]\setminus \{i\}$ if its actions are updated via the following rule:
\begin{subequations}\label{deceptiveNES}
\begin{align} 
    x_i(t)&=u_i(t)+a_i\sin(\omega_i t)+ \delta_i(t) \sum_{k\in\mathcal{D}_i}a_{k}\sin(\omega_{k} t)\label{decx}\\ 
    \dot{u}_i(t)&=-\dfrac{2k_i}{a_i} J_i(x(t))\sin(\omega_i t)
\end{align}
where $\delta_i(t)$ is a tuning deceptive gain that satisfies $\sup_{t\ge 0} |\delta_i(t)|>0$.
\end{subequations}
\end{definition}
In addition to knowing the frequency $\omega_i$, any player who wants to deceive player $i$ must also know the amplitude $a_i$. While this seem more restrictive than the deception considered in \cite{tang2024deception}, which only required knowledge of $\omega_i$, it is in fact less restrictive since the NES strategy in \cite{tang2024deception} assumed $a_k=a$ for all $k\in[N]$.

\vspace{0.2cm}
In our setting, we assume there are $n$ deceptive players, and the set of deceptive player is given by $\mathcal{D}=\{z_1,...,z_n\}$. We say player $i\in\mathcal{D}$ is deceiving $n_i\in[N]$ players in $\mathcal{D}_i:=\{d_{i,1},...,d_{i,n_i}\}$ if the parameter $\delta_i$ is updated according to:
\begin{equation}\label{deltadyn}
    \dot{\delta}_i=\varepsilon\varepsilon_i (J_i(x)-J_i^\text{ref}),\quad\varepsilon_i>0,\quad i\in\mathcal{D},
\end{equation}
where $J_i^\text{ref}$ is player $i$'s desired reference cost. For non-deceptive players, we have that $\delta_i:=0$. The full model-free NES dynamics can thus be stated as follows for all players $i\in[N]$:
\begin{subequations}\label{decgamedyn}
     \begin{align}
     x_i&=\begin{dcases}
        u_i+a_i\sin(\omega_i t)+\delta_i \sum_{j\in\mathcal{D}_i}a_j\sin\left(\omega_{j}t\right) & \text{if } i \in \mathcal{D} \\
        u_i+a_i\sin(\omega_i t) & \text{else}
    \end{dcases}\label{decgame1}\\
    \dot{u}_i&=-\frac{2k_i}{a_i}J_i(x)\sin(\omega_i t),\label{decgame2}
 \end{align}
\end{subequations}
Now that the overall game dynamics have been established, we will present our main results.
\section{Main Results}
\subsection{Stability}
Let $\mathcal{K}_j$ represent the set of players who are deceptive to player $j$, i.e $\mathcal{K}_j=\{i\in\mathcal{D}: j\in\mathcal{D}_i\}$. Moreover, let $\delta=[\delta_{z_1},...,\delta_{z_n}]^\top$. We define $\overline{\mathcal{Q}}(\delta)\in\mathbb{R}^{N\times N}, \overline{\cb}(\delta)\in\mathbb{R}^N$ with entries given by 
\begin{subequations}
    \begin{align}    [\overline{\cq}(\delta)]_{i:}&=[\cq]_{i:}+\sum_{k\in \mathcal{K}_i}\delta_k [Q_i]_{k:}\\
    [\overline{\cb}(\delta)]_i&=[\cb]_i+\sum_{k\in \mathcal{K}_i}\delta_k \frac{m_i R_{\|}}{R_i R_k}.
    \end{align}
\end{subequations}
These quantities essentially represent the ``perturbed" pseudogradient that results from deception, which will become more intuitive when we obtain the averaged system in the stability proof. Before we proceed to the stability analysis, we define the \emph{stability-preserving set}:
\begin{equation}
    \Delta=\{\delta\in\mathbb{R}^n : -K\overline{\cq}(\delta) \text{ is Hurwitz}\}
\end{equation}
where $K=\text{diag}([k_1,...,k_N]^\top)$. We can then present our first result, which provides a somewhat useful ``lower bound" on $\Delta$:
\begin{lemma}\label{deltaapprox}
    If $|\delta|<1$, then $\delta\in\Delta$.  
\end{lemma}
\begin{proof}
    By \eqref{qi} we already know $[Q_i]_{kj}=0$ for $k\in\ck_i$ and $j\neq i$, thus for fixed $i$ we have
    \begin{align}
        \sum_{j\neq i}^N|[\overline{\cq}(\delta)]_{ij}|=\sum_{j\neq i}^N\frac{R_{\|}}{R_i R_j}=\frac{R_{\|}}{R_i \overline{R}_i}.
    \end{align}
    and 
    \begin{align}
        [\overline{\cq}(\delta)]_{ii}=[\cq]_{ii}+\sum_{k\in \mathcal{K}_i}\delta_k [Q_i]_{ki}\ge [\cq]_{ii}-|\delta|\frac{R_{\|}}{R_i \overline{R}_i}>\frac{R_{\|}}{R_i \overline{R}_i}.
    \end{align}
    The result then follows by the Gershgorin Circle Theorem. \qedwhite
\end{proof}
This estimate is a significant improvement compared to the results of \cite{tang2024deception}, which only guarantee that $\Delta$ contains a neighborhood of the origin. By exploiting the structure of the duopoly, we are now able to say that this neighborhood at least contains the unit ball. To characterize when the deceptive players are able to properly achieve their desired payoffs via deception, we introduce the notion of \emph{attainability} as was presented in \cite{tang2024deception}:
\begin{definition}\label{jattaindef}
    A vector $J^{\text{ref}}=[J_{z_1}^{\text{ref}},...,J_{z_n}^{\text{ref}}]^\top$ is said to be attainable if there exists $\delta^*\in\Delta$ such that:
    \begin{enumerate}
        \item $J_{z_k}(-\overline{\cq}(\delta^*)^{-1}\overline{\cb}(\delta^*))=J_{z_k}^{\text{ref}},\quad\forall~~k\in [n]$.
        \item The matrix $\Lambda(\delta^*)\in\mathbb{R}^{n\times n}$ with $[\Lambda(\delta^*)]_{jk}=\nabla_j \xi_k(\delta^*)$ is Hurwitz, where $\xi_k:\mathbb{R}^{n}\to\mathbb{R}$ is given by $\xi_k(\delta):=\varepsilon_{z_k} J_{z_k}(-\overline{\cq}(\delta)^{-1}\overline{\cb}(\delta))$.
    \end{enumerate}
 We let $\Omega\subset\mathbb{R}^n$ denote the set of all \emph{attainable} vectors $J^{\text{ref}}=[J_{z_1}^{\text{ref}},...,J_{z_n}^{\text{ref}}]^\top$.
\end{definition}
With this definition at hand, we can now state the main result of this paper. The following theorem characterizes the stability properties of the NES dynamics with deception:
%
% \begin{theorem}
% %
% Consider the NES seeking dynamics  \eqref{deltadyn} and \eqref{decgamedyn} with $J^\text{ref}\in\Omega$, $J_i$ of the form \eqref{jcost} and $\omega_i$ satisfying Assumption \ref{assumpw} for all $i\in[N]$. Then there exists $\varepsilon^*>0$ such that for all $\varepsilon\in (0, \varepsilon^*)$, there exists $a^*>0$ such that for $a_1,...,a_N\in (0, a^*)$ there exists $\omega^*>0$ such that for all $\omega>\omega^*$ the state $\zeta(t):=[u(t)\quad\delta(t)]^\top$ converges exponentially to a $\mathcal{O}(\frac{1}{\omega}+\max_i a_i)$-neighborhood of a point $\zeta^*:=[u^*\quad\delta^*]^\top$, provided $|\zeta(0)-\zeta^*|$ is sufficiently small.
% %
% \end{theorem}

\begin{theorem}
Consider the NES seeking dynamics with $J^\text{ref}\in\Omega$. Then there exists $\varepsilon^*>0$ such that for all $\varepsilon\in (0, \varepsilon^*)$, there exists $a^*>0$ such that for $a_1,...,a_N\in (0, a^*)$ there exists $\omega^*>0$ such that for all $\omega>\omega^*$ the state $\zeta(t):=[u(t)\quad\delta(t)]^\top$ converges exponentially to a $\mathcal{O}(\frac{1}{\omega}+\max_i a_i)$-neighborhood of a point $\zeta^*:=[u^*\quad\delta^*]^\top$, provided $|\zeta(0)-\zeta^*|$ is sufficiently small.
\end{theorem}

\begin{proof}
    To analyze the system, let ${\mu}(t)=x-u$, where $x, u$ are given in \eqref{decgamedyn}. In other words, $\mu(t)$ is the vector of sinusoids. We apply the time scale transformation $\tau=\omega t$, and denote $\tilde{\mu}(\tau)=\mu(\tau/\omega)$ and $T=2\pi \times \lcm\{{1}/{\overline{\omega}_1}, {1}/{\overline{\omega}_2}, ..., {1}/{\overline{\omega}_N}\}$. With standard averaging theory \cite{khalil}, we can compute the average dynamics of system \eqref{decgamedyn}, whose state we denote as $\tilde{u}\in\mathbb{R}^N$:
    \begin{align}
    \frac{\partial{\tilde{u}}_i}{\partial \tau}&=\frac{1}{\omega T}\int_{0}^{T}-\frac{2k_i}{a_i}J_i(\tilde{u}+\tilde{\mu}(\tau))\sin(\overline{\omega}_i \tau)d\tau\notag\\
    &=-\frac{2k_i}{a_i\omega T}\int_{0}^{T}\biggl(J_i(\tilde{u})+\tilde{\mu}(\tau)^\top\nabla J_i(\tilde{u})+\frac12\tilde{\mu}(\tau)^\top Q_i \tilde{\mu}(\tau)\biggl)\sin(\overline{\omega}_i \tau)d\tau\\
    &=-\frac{2k_i}{a_i\omega T}\int_{0}^{T}\sin(\overline{\omega}_i \tau)\tilde{\mu}(\tau)^\top (Q_i \tilde{u}+b_i)d\tau\notag\\
    &=-\frac{k_i}{\omega}\left([Q_i]_{i:}\tilde{u}+[b_i]_i+\sum_{k\in\mathcal{K}_i}\tilde{\delta}_k ([Q_i]_{k:}\tilde{u}+[b_i]_k)\right).\label{decavgsystem}
\end{align}
By combining all $u_i$'s, we obtain the average system
\begin{equation}
    \frac{\partial \tilde{u}}{\partial\tau}=\frac{1}{\omega}\left(-K\overline{\cq}(\tilde{\delta})\tilde{u}-K\overline{\cb}(\tilde{\delta})\right).
\end{equation}
We can apply the same technique on \eqref{decgamedyn} for $i\in\mathcal{D}$:
\begin{align}
        \dfrac{\partial \tilde{\delta}_i}{\partial \tau}&=\dfrac{\varepsilon}{\omega}\frac{1}{T}\int_{0}^{T}\varepsilon_i\biggl( J_i(\tilde{u})-J_i^{\text{ref}}+\tilde{\mu}(\tau)^\top\nabla J_i(\tilde{u})+\frac12\tilde{\mu}(\tau)^\top Q_i\tilde{\mu}(\tau)\biggl)d\tau\label{thm1deltat}\\
        &=\dfrac{\varepsilon}{\omega}\varepsilon_i \left(J_i(\tilde{u})-J_i^{\text{ref}}+\mathcal{P}_i([a_1,...,a_N]^\top)\right),\quad i\in \mathcal{D},\label{thm1delta}
\end{align}
where $P_i:\mathbb{R}^N\to\mathbb{R}$ is a quadratic function satisfying $\mathcal{P}_i(a)=0$, and can thus be treated as an $\mathcal{O}(a)$ perturbation on compact sets, where $a=[a_1,...,a_N]^\top$. We recall that $\mathcal{D}=\{z_1,...,z_n\}$. By setting $\tau^*=\varepsilon \tau$ and ${J}^*(\tilde{u})=[J_{z_1}(\tilde{u}),...,J_{z_n}(\tilde{u})]^\top$, the entire system can be represented as
\begin{equation}\label{thm1spf}
        \begin{bmatrix}
            \varepsilon\dfrac{\partial \tilde{u}}{\partial \tau^*}\\
            \dfrac{\partial \tilde{\delta}}{\partial \tau^*}
        \end{bmatrix}=\dfrac{1}{\omega}\begin{bmatrix}
            -K\overline{\cq}(\tilde{\delta})\tilde{u}-K\overline{\cb}(\tilde{\delta})\\
            \text{diag}(\varepsilon_{z_1},...,\varepsilon_{z_n})\left({J}^*(\tilde{u})-J^{\text{ref}}\right)
        \end{bmatrix}+\mathcal{O}(a).
    \end{equation}
    If we disregard the perturbation in \eqref{thm1spf}, the resulting system is a singularly perturbed system with quasi steady state $h(\tilde{\delta})=-\overline{\cq}(\tilde{\delta})^{-1}\overline{\cb}(\tilde{\delta})$ and reduced dynamics given by
     \begin{equation}
        \dfrac{\partial \tilde{\delta}}{\partial \tau^*}=\dfrac{1}{\omega}\text{diag}(\varepsilon_{z_1},...,\varepsilon_{z_n})\left({J}^*(-\overline{\cq}(\tilde{\delta})^{-1}\overline{\cb}(\tilde{\delta}))-J^{\text{ref}}\right).\label{thm1red}
    \end{equation}
    Since $J^\text{ref}\in\Omega$, it follows that \eqref{thm1red} has an exponentially stable equilibrium point $\delta^*\in\Delta$. Moreover, by denoting $y=\tilde{u}-h(\tilde{\delta})$, we obtain the boundary layer system
    \begin{equation}
        \frac{\partial y}{\partial \tau}=-K\overline{\cq}(\tilde{\delta})y
    \end{equation}
    where the origin is exponentially stable uniformly in any compact set contained in $\Delta$. Hence the unperturbed system \eqref{thm1spf} has an exponentially stable equilibrium point $\zeta^*=[u^*\quad\delta^*]^\top$ where $u^*=-\overline{\cq}(\delta^*)^{-1}\overline{\cb}(\delta^*)$. By standard robustness results for systems with small additive perturbations, we can find $a^*>0$ such that for $\max_i a_i\in(0, a^*)$, $\tilde{\zeta}$ converges exponentially to a $\mathcal{O}(\max_i a_i)$ neighborhood of $\zeta^*$ provided $|\zeta(0)-\zeta^*|$ sufficiently small. Then, by standard averaging results for ODEs \cite[Thm 10.4]{khalil} we prove the claim for $\omega$ sufficiently large.\qedwhite
\end{proof}
Even though this proof is similar to that of the main result in \cite{tang2024deception}, we include it here to account for the generalization of allowing distinct $k_i$ and $a_i$. While this result is quite useful, one main concern is computing the set $\Omega$. A relatively straightforward method for computing $\Omega$ when $|\mathcal{D}|=|\mathcal{D}_{z_1}|=1$ is presented in \cite{tang2024deception}, but in general it is a challenging problem. One promising approach is to use Lemma \ref{deltaapprox} to evaluate a ``lower bound" for $J_i(-\cq(\delta)^{-1}\cb(\delta))$ and then use numerical methods to approximate $\Lambda(\delta)$ for $|\delta|<1$.
\subsection{Intuition Behind Deception}
In \cite{tang2024deception}, it was shown that deception via \eqref{decgamedyn} essentially transforms the duopoly into a deceptive game with deceptive costs that manipulate the victim's seeking dynamics into misinterpreting their sales function $s_i(x)$, so it is of interest to generalize this intuition to the $N$-player oligopoly. We first recall the definition of a \emph{deceptive game} from \cite{tang2024deception}:
\begin{definition}\label{defdeceptivegames}
Given a non-cooperative game $\{J_i\}_{i\in[N]}$, we say that $\{\tilde{J}_i\}_{i\in[N]}$ is a deceptive game if there exists a nonempty subset  $\mathcal{D}\subset[N]$ such that for each $i\in \mathcal{D}$ there exists a function $\sigma_i:\mathbb{R}^{N-1}\to\mathbb{R}$, a nonempty set $\mathcal{K}_i\subset[N]\setminus\{i\}$ and scalars $\delta_{k}\neq 0~~\forall k\in\mathcal{K}_i$, such that $\tilde{J}_i$ satisfies:
%
% \begin{equation}\label{deceptivecost}
% \tilde{J}_i(x)=\left\{\begin{array}{ll}
% J_i(x)+\sum_{k\in\mathcal{K}_i}\delta_{k}\int_{0}^{x_i}\nabla_{k} 
%  J_{i} (y)dy_i &\text{if}~i\in \mathcal{D},\\
% J_i(x) &\text{if}~i\not\in \mathcal{D},
% \end{array}\right.
% \end{equation}
\begin{equation}\label{deceptivecost}
    \tilde{J}_i(x)=J_i(x)+\sigma_i(x_{-i})+\sum_{k\in\mathcal{K}_i}\delta_{k}\int_{0}^{x_i}\nabla_{k}  J_{i} (y, x_{-i})dy,
\end{equation}
for all $x\in\mathbb{R}^N$. If $i\not\in \mathcal{D}$, then $\tilde{J}_i(x)=J_i(x)$. 
\end{definition}
In other words, this definition captures the effect of deception by framing the emerging behavior of the system as a standard Nash equilibrium-seeking problem parameterized by costs $\{\tilde{J}_i\}_{i\in[N]}$. That is, the game with costs $\{{J}_i\}_{i\in[N]}$ under the deceptive NES dynamics \eqref{decgamedyn} will exhibit the same asymptotic behavior as a standard non-cooperative game with costs $\{\tilde{J}_i\}_{i\in[N]}$ and players implementing the deception-free NES dynamics \eqref{extsc0}. We can then use \eqref{deceptivecost} to compute a deceptive game for the oligopoly market scenario:
\begin{align}
    \tilde{J}_i(x)&=J_i(x)+\sigma_i(x_{-i})-\sum_{k\in\ck_i}\delta_k\int_{0}^{x_i} \frac{R_{\|}}{R_i R_k}(y-m_i) dy\\
    &=-s_i(x)(x_i-m_i)+\sigma_i(x_{-i})-\sum_{k\in\ck_i}\frac{\delta_k R_{\|}x_i}{2R_i R_k}\left(x_i-2m_i\right).
\end{align}
If we set
\begin{equation}
    \sigma_i(x_{-i}):=-\frac{R_{\|}m_i^2}{2R_i}\sum_{k\in\ck_i}\frac{\delta_k}{R_k}
\end{equation}
we finally obtain
\begin{align}
    \tilde{J}_i(x)&=-\left(s_i(x)+\frac{R_{\|}x_i}{2R_i}\sum_{k\in\ck_i}\frac{\delta_k}{R_k}\right)(x_i-m_i)+(x_i-m_i)\frac{R_{\|}m_i}{2R_i}\sum_{k\in\ck_i}\frac{\delta_k}{R_k}\\
    &=-\left(s_i(x)+(x_i-m_i)\frac{R_{\|}}{2R_i}\sum_{k\in\ck_i}\frac{\delta_k}{R_k}\right)(x_i-m_i).\label{inflated}
\end{align}
When comparing \eqref{inflated} with \eqref{jcost}, we can observe that player $i$ now has a $\delta$-inflated sales function, which is highly reminiscent of the duopoly analysis from \cite{tang2024deception}. If we reasonably assume that $x_i>m_i$, then $\frac{R_{\|}}{2R_i}\sum_{k\in\ck_i}\frac{\delta_k}{R_k}>0$ implies that player $i$'s NES dynamics behave as if their sales are greater than they really are. This will result in player $i$ increasing their price $x_i$ to increase their payoff, even though in reality it might harm their profits. Similarly, if $\frac{R_{\|}}{2R_i}\sum_{k\in\ck_i}\frac{\delta_k}{R_k}<0$, player $i$'s strategy will behave as if their sales are lower than their true value. This will result in player $i$ decreasing their price $x_i$, which could potentially harm their profits. 
\begin{remark}
Another useful observation is the term $\frac{\delta_k}{R_k}$, which indicates that a firm with a more desirable product will wield more influence as a deceiver.
\end{remark}
Although the ``deceptive game" view adds some intuition behind how deception affects the players' behavior, one of the main drawbacks is that the variability of $\sigma_i$ may lead to conflicting interpretations. Hence, we present an alternate and somewhat consistent viewpoint that captures the effect of deception on player $i$'s estimate of their gradient $\nabla_i J_i$.  From \eqref{deceptivecost} we have:
\begin{align*}
    \nabla_i \tilde{J}_i(x)&=\nabla_i J_i(x)+\sum_{k\in\ck_i}\delta_k \nabla_k J_i(x)\\
    &=\left([Q_i]_{i:}+\sum_{k\in\ck_i}\delta_k [Q_i]_{k:}\right)x+\left([b_i]_{i}+\sum_{k\in\ck_i}\delta_k [b_i]_{k}\right)\\
    &=[\overline{\cq}(\delta)]_{i:}x+[\overline{\cb}(\delta)]_i.
\end{align*}
However, note that $[\overline{\cq}(\delta)]_{ij}=[Q_i]_{ij}$ for $i\neq j$, and we also have
\begin{align}
    [\overline{\cq}(\delta)]_{ii}&=[Q_i]_{ii}+\sum_{k\in\ck_i}\delta_k [Q_i]_{ki}\notag\\
    &=\frac{2R_{\|}}{R_i\overline{R}_i}-\sum_{k\in\ck_i}\delta_k\frac{R_{\|}}{R_i R_k}\notag\\
    &=\frac{R_{\|}}{R_i}\left(\frac{2}{\overline{R}_i}-\sum_{k\in\ck_i}\frac{\delta_k}{R_k}\right)\label{intq}
\end{align}
and 
\begin{align}    \left([b_i]_{i}+\sum_{k\in\ck_i}\delta_k [b_i]_{k}\right)&=-\frac{m_i R_{\|}}{R_i\overline{R}_i}-\frac{S_d R_{\|}}{R_i}+\sum_{k\in\ck_i}\delta_k\frac{m_i R_{\|}}{R_i R_k}\notag\\
&=-\frac{m_i R_{\|}}{R_i}\left(\frac{1}{\overline{R}_i}-\sum_{k\in\ck_i}\frac{\delta_k}{R_k}\right)-\frac{S_d R_{\|}}{R_i}\label{intb}
\end{align}
Both cases seem to indicate that when player $i$ is being attacked, their NES dynamics ``learn" that the total desirability of the other players' products is given by $\frac{1}{\overline{R}_i}-\sum_{k\in\ck_i}\frac{\delta_k}{R_k}$ instead of $\frac{1}{\overline{R}_i}$.

This means that whenever $\sum_{k\in\ck_i}\frac{\delta_k}{R_k}>0$ holds, player $i$ believes the other player's products to be less desirable than they really are, which would lead to player $i$ increasing their price $x_i$. Similarly, when 
$\sum_{k\in\ck_i}\frac{\delta_k}{R_k}<0$, player $i$'s NES dynamics will overestimate the desirability of the other players' products, leading to player $i$ lowering their price. Alternatively, one can also make the following observation:
\begin{equation}
    \frac{1}{\overline{R}_i}-\sum_{k\in\ck_i}\frac{\delta_k}{R_k}=\sum_{j\not\in\ck_i\cup\{i\}}\frac{1}{R_j}+\sum_{k\in\ck_i}\frac{1-\delta_k}{R_k}\label{dec_rk}
\end{equation}
which suggests something slightly more intricate. In particular, although \eqref{intq} and \eqref{intb} provide some insight into how the players in $\ck_i$ collectively alter the beliefs of player $i$, equation \eqref{dec_rk} indicates the degree to which each deceptive player affects player $i$'s estimate of the desirability of that deceptive player's product. Just as we have seen before, the term $\frac{1}{R_k}$ implies that the desirability of a firm's product can amplify their ability to deceive.

So far, we have observed that when firms in an oligopoly implement the deceptive NES strategy \eqref{decgamedyn}, the asymptotic behavior aligns with that of a game parameterized by costs $\{\tilde{J}_i\}_{i\in[N]}$ with firms implementing the deception-free NES dynamics \eqref{extsc0}, but it is also of interest to ask if the new ``deceptive" equilibrium point is actually a Nash equilibrium of the deceptive game $\{\tilde{J}_i\}_{i\in[N]}$. The following result provides some sufficient conditions to answer this question. We present the proof for completeness:
\begin{theorem}
    Consider the $N-$player oligopoly with costs $J_i$ of the form \eqref{jcost}. Let $\delta\in \{\delta\in\mathbb{R}^n: 0<|\delta|<2 \emph{ and } \overline{\cq}(\delta) \emph{ is invertible} \}$ and let $\{\tilde{J}_i\}_{i\in[N]}$ be a corresponding deceptive game. Then, the point $u^*=-\overline{\cq}(\delta)^{-1}\overline{\cb}(\delta)$ is a Nash equilibrium of the deceptive game $\{\tilde{J}_i\}_{i\in[N]}$.
\end{theorem}
\begin{proof}
Since the pseudogradient of the deceptive game satisfies $\mathcal{G}(x)=\overline{\cq}(\delta)x+\overline{\cb}(\delta)$, it is easy to verify that if $\overline{\cq}(\delta)$ is invertible, the first order condition \cite{bacsar1998dynamic} for $u^*$ to be a Nash equilibrium is satisfied. We also have
\begin{align}
        [\overline{\cq}(\delta)]_{ii}=[\cq]_{ii}+\sum_{k\in \mathcal{K}_i}\delta_k [Q_i]_{ki}\ge [\cq]_{ii}-|\delta|\frac{R_{\|}}{R_i \overline{R}_i}>0,
\end{align}
since $|\delta|<2$. Thus, the second order condition \cite{bacsar1998dynamic} is satisfied, which implies $u^*$ is a Nash equilibrium of the deceptive game $\{\tilde{J}_i\}_{i\in[N]}$.\qedwhite
\end{proof}
\section{Numerical Results}
 \begin{figure}[t!]
  \centering
    \includegraphics[width=0.7\textwidth]{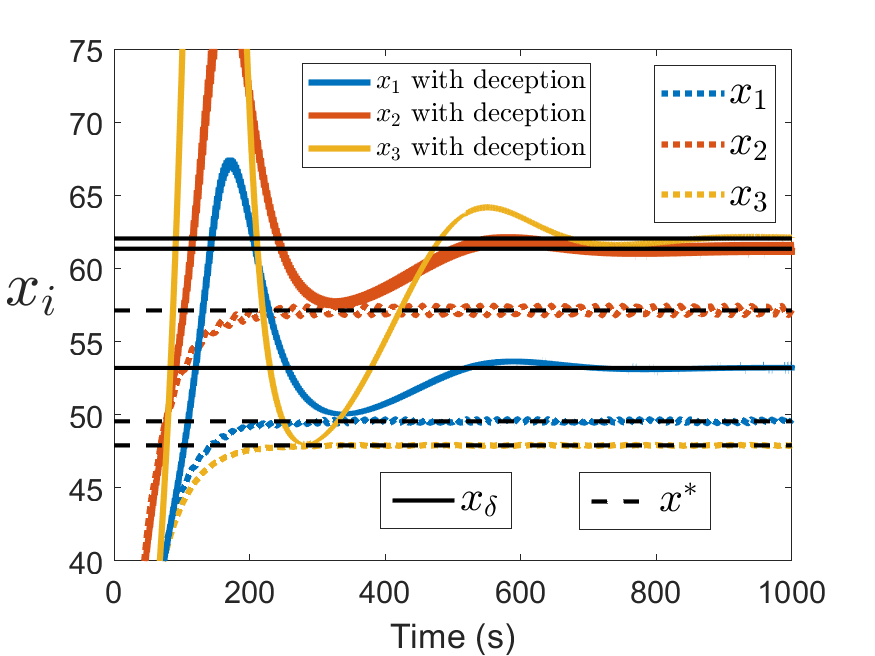}
    \caption{\small The price convergence for a three player oligopoly, where we use $x^*$ and $x_\delta$ to denote the true Nash equilibrium and the deceptive Nash equilibrium, respectively.}\label{figex}
    \vspace{-0.5cm}
\end{figure}
To further illustrate our theoretical results, we present a three-firm oligopoly example with parameters $R_1=0.67,~R_2=0.36,~R_3=0.8, m_1=20,~m_2=29,~m_3=30$ and $S_d=100$. It is simple to check that this game has Nash equilibrium $x_0=[49.55, 57.13, 47.9]^\top$ and steady state costs $J_1(x_0)=-950.7, J_2(x_0)=-1092$, and $J_3(x_0)=-239.2$. For the NES dynamics, the players use $a_1=0.04, k_1=0.02, \omega_1=6346,
a_2=0.03,
k_2=0.019, \omega_2=4089, a_3=0.05$, $k_3=0.22$, and $\omega_3=6115$. Now we let player 1 be deceptive to player 3, where player 1 tunes $\delta_1$ according to \eqref{deltadyn} with $\varepsilon=10^{-4}$, $\varepsilon_1=1$ and $J_1^\text{ref}=-1200$ (recall that for our algorithm, we treat $J_i$ as a \emph{cost} to be \emph{minimized}). Moreover, we have
\begin{subequations}
    \begin{align}
        \overline{\cq}(\delta)&=\begin{bmatrix}
            2.18 & -0.75 & -0.34\\
            -0.75 & 2.76 & -0.63\\
            -0.34 & -0.63&2.18
        \end{bmatrix}-\delta\begin{bmatrix}
            0 & 0 & 0\\
            0 & 0 & 0\\
            0 & 0&0.34
        \end{bmatrix}\\
        \overline{\cb}(\delta)&=\begin{bmatrix}
            -48.82 \\ -90.34 \\ -51.65
        \end{bmatrix}+\delta\begin{bmatrix}
            0 \\ 0 \\ 10.14
        \end{bmatrix}
    \end{align}
\end{subequations}
 \begin{figure}[t!]
  \centering
    \includegraphics[width=0.52\textwidth]{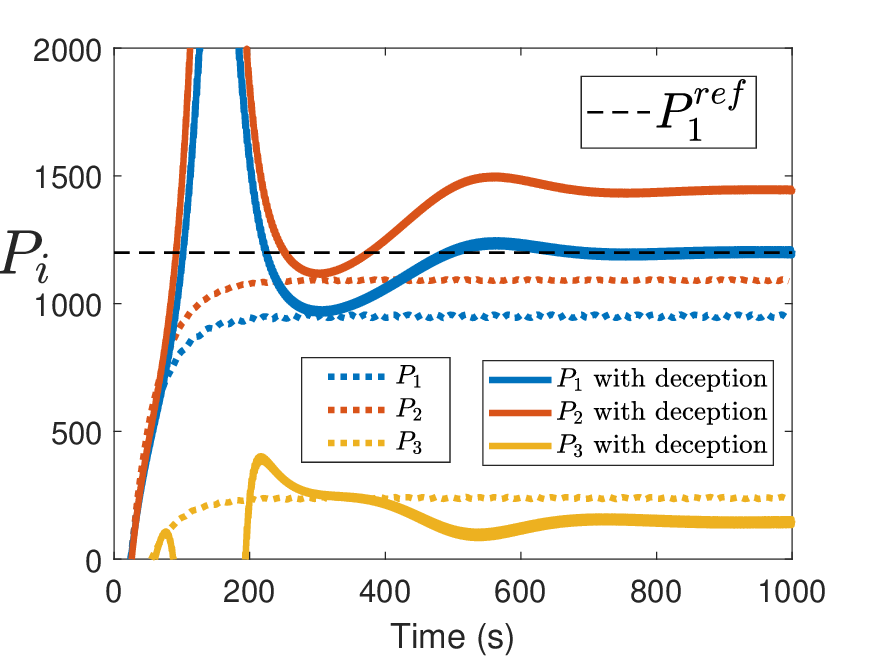}\hspace{-0.5cm}\includegraphics[width=0.52\textwidth]{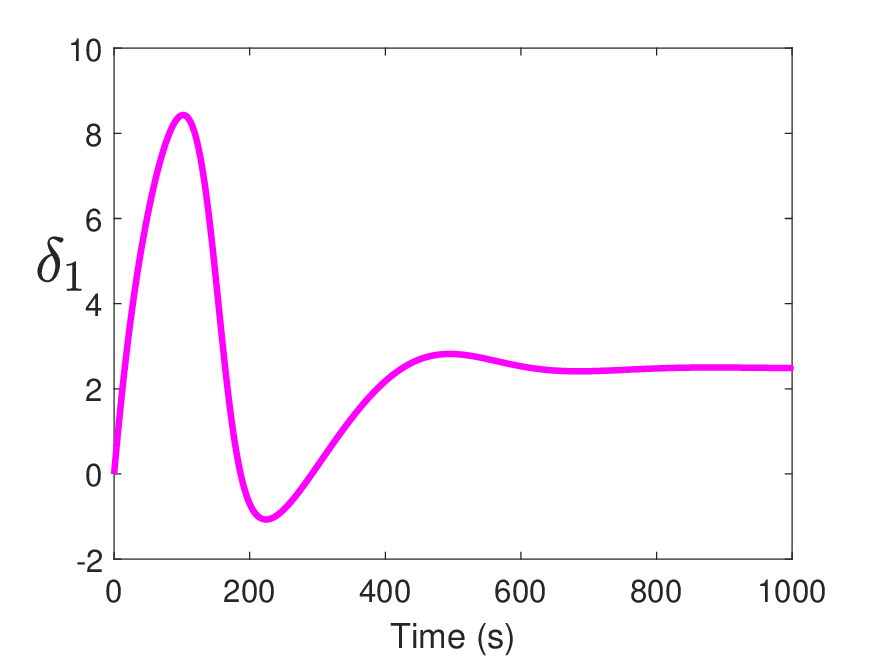}
    \caption{\small Left: The profit convergence for a three player oligopoly, where $P_i=-J_i$ and $P_1^\text{ref}=-J_1^\text{ref}$. Right: The corresponding trajectories of $\delta_1$ when player 1 deceives player 3.}\label{figex}
\end{figure}

\noindent
By numerically solving $J_1(-\overline{\cq}(\delta^*)^{-1}\overline{\cb}(\delta^*))=J_1^\text{ref}$ we obtain $\delta^*=2.486\in\Delta$. Moreover, we have $\frac{\partial}{\partial\delta}J_1(-\overline{\cq}(\delta)^{-1}\overline{\cb}(\delta))|_{\delta=\delta^*}=-190<0$, so $J_1^\text{ref}\in\Omega$. For a more intuitive illustration, we will only plot the \emph{profits} $P_i=-J_i$. As we see in the plots, player $1$ is able to force the dynamics to converge to a neighborhood of a point $x_\delta$ that achieves $J_1(x_\delta)=J_1^\text{ref}$. Moreover, we notice that the deception mechanism causes player 3 to increase their price $x_3$ significantly, which is consistent with our observations from \eqref{inflated}, \eqref{intq} and \eqref{intb} since $\delta_1^*>0$.
\section{Conclusion}
In this work, we performed a comprehensive study that applied the deception mechanism from \cite{tang2024deception} to an $N$-player \emph{oligopoly market}. By leveraging the structure of the basic two-market duopoly, we establish stability for our proposed deception mechanism in which $N$ players have different gains and exploration amplitudes for their NES dynamics. Moreover, the diagonal dominance property of the duopoly game allows us to derive an improved estimate on the ``stability-preserving set" and quantify when the deceptive Nash equilibrium is actually a NE for the corresponding oligopoly class of deceptive games. It was also shown that the deception mechanism can be economically interpreted as an artificial inflation (or deflation) of the desirability of the deceivers' product learned by the victim. Future research directions will characterize how the structure of the interaction graph between players (encoded on how the cost function $J_i$ depends on the actions of the other players) affects the emerging deceptive Nash equilibrium and its stability properties. It is also of interest to study stochastic and hybrid (involving continuous-time and discrete-time dynamics) deception mechanisms.
\bibliographystyle{splncs04}
\bibliography{mybibliography}

\end{document}